\documentclass[10pt,conference]{IEEEtran}
\usepackage{verbatim}
\usepackage{amsfonts,amsmath,amssymb}
\usepackage{graphicx}
\begin{document}
%-------------------------------------------------------------------------
\newtheorem{theorem}{Theorem}
\newtheorem{lemma}{Lemma}
\newtheorem{remark}{Remark}
\newcommand{\Csum}{C_{\text{sum}}}
\newcommand{\snr}{P}
\newcommand{\prob}{\text{Pr}}
%-------------------------------------------------------------------------
% paper title
\title{{\small Proceedings of ITA Workshop, San Diego, CA, Jan-Feb, 2008.}\\Sum Capacity of the Gaussian Interference Channel in the Low Interference Regime}
%-------------------------------------------------------------------------
% author names and affiliations
% use a multiple column layout for up to three different
% affiliations
\author{
\authorblockN{V. Sreekanth Annapureddy and Venugopal Veeravalli\\}
\authorblockA{Coordinated Science Laboratory \\
Department of Electrical and Computer Engineering\\
University of Illinois at Urbana-Champaign \\
vannapu2@uiuc.edu, vvv@uiuc.edu}
}
%-------------------------------------------------------------------------
% make the title area
\maketitle
%-------------------------------------------------------------------------
\begin{abstract}
New upper bounds on the sum capacity of the two-user Gaussian interference channel are derived. Using these bounds, it is shown that treating interference as noise achieves the sum capacity if the interference levels are below certain thresholds.
\end{abstract}
%-------------------------------------------------------------------------
\section{Introduction}
Interference is a fundamental issue in the design of  communication networks, particularly wireless networks. Unlike thermal noise, interference has a structure since it is generated by other users. Can this structure be exploited to decrease the uncertainty and thus improve the performance of the communication network? If so, what are the optimal signalling strategies? In this paper, we show that exploiting the structure of the interference in a two-user Gaussian interference channel does not improve the overall system throughput  in the {\em low interference regime}. In other words, one can treat interference as noise and can still achieve the maximum possible throughput, if the interference levels are below certain thresholds.

The capacity region of the two-user Gaussian interference channel is known in the {\em strong interference} setting \cite{Carleial1975, HK1981, Sato1981}, where it is shown that each user can decode the information transmitted to the other user, and in the  trivial case when there is no interference. The sum capacity of the interference channel is known for the one-sided interference channel (also called the Z-Channel) \cite{Kramer2004,Sason2004,OneBit2007}, where treating interference as noise achieves the sum capacity, and the degraded interference channel \cite{Costa1985},\cite{Sason2004}, where one user treats interference as noise and the other user does interference cancelation.

Establishing the capacity region for a general two-user Gaussian interference channel still remains an open problem. The best known achievable strategy is the Han-Kobayashi scheme \cite{HK1981}, where each user splits the information into private and common parts. The common messages are decoded at both the receivers, thereby reducing the level of interference. Although Chong, Motani and Garg have recently derived a simple representation of the Han-Kobayashi achievable region \cite{CMG2007}, it still remains formidable to compute.

In \cite{OneBit2007}, the capacity region of a general two-user Gaussian interference channel is determined to {\em within one bit} by comparing a special case of the Han-Kobayashi scheme to the outer bounds derived in \cite{Kramer2004} and \cite{OneBit2007}. The concept of a {\em genie-aided} channel is used in deriving the outer bounds, where the receivers are provided with side information by a genie. The side information is chosen in such a way as to facilitate the computation of the capacity region of the genie-aided channel, which is an obvious outer bound to the capacity region of the interference channel.

In this paper, we tighten the outer bound on the sum capacity derived in \cite{OneBit2007}. In a low interference regime, we establish the existence of a genie, which results in a genie-aided channel whose sum capacity can be computed, and yet does not improve upon the sum capacity of the interference channel. Thus, we establish the sum capacity of the two-user Gaussian interference channel in this  {\em low interference} regime, where the interference parameters are below certain pre-computable thresholds. In this regime, we further establish that it is optimal for the receivers to employ single user decoders that treat the interference as noise.

%-------------------------------------------------------------------------
%-------------------------------------------------------------------------
\section{Interference Channel Model}
The two-user Gaussian interference channel that we study in this paper is in the standard form \cite{Carleial1978}, \cite{HK1981}. Over one symbol period the channel is described by
 \begin{equation}
 \begin{split}
Y_1 = & \ X_1 + h_{12}X_2 + Z_1 \\
Y_2 = & \ X_2 + h_{21}X_1 + Z_2
\label{channel_model}
\end{split}
\end{equation}
with inputs $X_1, X_2$, and corresponding outputs $Y_1, Y_2$. The receiver noise terms $Z_1$ and $Z_2$ are assumed to be independent, zero-mean, unit variance Gaussian random variables, and the interference parameters $h_{12}$ and $h_{21}$ are assumed to be real numbers. The transmit power constraints on users 1 and 2 are $\snr_1$ and $\snr_2$, respectively. The noise terms are assumed to be independent and identically distributed (i.i.d.) in time.

For each user $i$, let the message index ($m_i$) be uniformly distributed over $\{1,2,\ldots,2^{nR_i}\}$ and $\mathcal{C}_i(n)$ be a code consisting of an encoding function $X_i^n: \{1,2,\ldots,2^{nR_i}\} \rightarrow \mathbb{R}^n$ satisfying the power constraint \[ ||X_i^n(m_i)||^2 \leq nP_i, \forall m_i \in \{1,2,\ldots,2^{nR_i}\} \] and a decoding function $g_i: \mathbb{R}^n \rightarrow \{1,2,\ldots,2^{nR_i}\}$. The corresponding probability of decoding error $\lambda_i(n)$ defined as $\prob[m_i \neq g_i(Y_i^n)]$. A rate pair $(R_1,R_2)$ is said to be achievable if there exists a sequence of codes $\{ \mathcal{C}_1(n),\mathcal{C}_2(n) \}$ such that the error probabilities $\lambda_1(n)$ and $\lambda_2(n)$ go to zero as $n$ goes to infinity.

\subsection{Notation}
% maybe this should be introduced when needed??
%All variants of $Z$ ($W$, $\hat{Z}$, etc.) denote unit variance white Gaussian noise.
The variables $S_1$ and $S_2$ denote the side information given to receivers $1$ and $2$, respectively. The variables $X_{1G}$ and $X_{2G}$ denote zero-mean Gaussian random variables with variances $P_1$ and $P_2$, respectively. The variables $Y_{1G},S_{1G},Y_{2G}$ and $S_{2G}$ denote the Gaussian outputs and side information that result when the channel inputs are Gaussian, i.e., when  $X_1 = X_{1G}$ and $X_2 = X_{2G}$.
%-------------------------------------------------------------------------
%-------------------------------------------------------------------------
\section{Symmetric Interference Channel}
The essential ideas and results of this paper are captured in the {\em symmetric }interference channel, for which  $P_1 = P_2 = P$ and $h_{12} = h_{21} = h$. For this channel we shall establish the following result.
\begin{theorem}
For the symmetric interference channel, if the interference parameter $h$ satisfies the condition
\begin{eqnarray}
|h +h^3P| \leq .5
\label{condition_mainresult}
\end{eqnarray}
then treating interference as noise achieves the sum capacity, which is given by
\[
\Csum = \log\left(1+\frac{P}{1+h^2P}\right)
\]
\label{mainresult}
\end{theorem}
%-------------------------------------------------------------------------
\subsection{Existing Bounds}
A natural way to deal with  interference between users is to treat interference as noise if the interference is weak, and to orthogonalize the users if the interference is moderate.  Therefore, the sum capacity of the symmetric interference channel is easily seen to be lower bounded as:
\begin{equation}
\Csum \geq \log\left(1 + \frac{P}{1+h^2P}\right)
\label{sumrate_tin}
\end{equation}
\begin{equation}
\Csum \geq \log\left(1 + 2P\right)
\label{sumrate_odm}
\end{equation}
The optimality of either of these simple strategies is not clear and has not been established previously. More sophisticated strategies such as splitting power into private and common messages, which require multiuser decoders and knowledge of the interfering users' codebooks, have been proposed  by Han and Kobayashi \cite{HK1981}.  A simplified version of the Han-Kobayashi strategy was  recently shown to produce an achievable region that is  within one bit of the capacity region \cite{OneBit2007}.

Regarding upper bounds on the sum capacity, genie-based arguments have been used in  \cite{Kramer2004, OneBit2007} to obtain the following:
\begin{equation}
\Csum \leq \log\left(1 + h^2P + \frac{P}{1+h^2P}\right)
\label{upperbound_onebit}
\end{equation}
\begin{equation}
\Csum \leq \frac{1}{2}\log\left(1 + P\right) +  \frac{1}{2}\log\left(1 + \frac{P}{1+h^2P}\right)
\label{upperbound_kramer}
\end{equation}
The upper bound given in \eqref{upperbound_onebit}, which we refer to as the One-Bit bound,  is asymptotically tight in the low interference regime \cite{OneBit2007}. The upper bound given in \eqref{upperbound_kramer}, the Z-Channel bound, is asymptotically tight in the moderate interference interference regime. (See Fig.~\ref{fig_mainresult}.)

In this paper, the upper bound given in  (\ref{upperbound_onebit}) is tightened to establish Theorem~\ref{mainresult}. Furthermore, the upper bound given in (\ref{upperbound_kramer}) is shown to be a special case of Theorem~\ref{sum_capacity_asym}, which extends Theorem~\ref{mainresult} to the asymmetric interference channel.

\begin{figure}
\centering
\includegraphics[width=.5\textwidth]{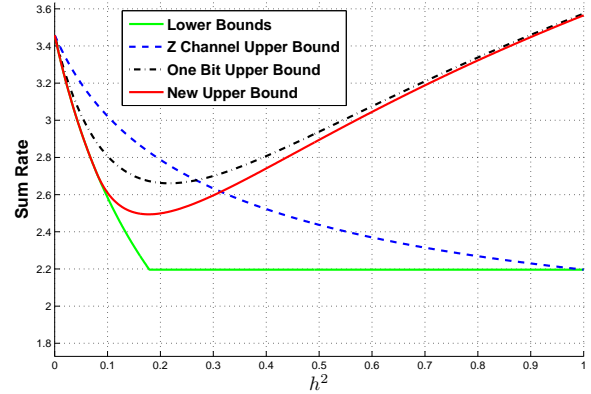}
\caption{Bounds on the sum capacity, $P = 10$ dB}
\label{fig_mainresult}
\end{figure}
%-------------------------------------------------------------------------
\subsection{Proof of Theorem~\ref{mainresult}}
To prove Theorem~\ref{mainresult}, we need to establish an upper bound on $\Csum$ that matches the lower bound given in (\ref{sumrate_tin}), when condition (\ref{condition_mainresult}) is satisfied. As in \cite{Kramer2004, OneBit2007}, our upper bound is based on a genie giving side information to the receivers. The genie needs to be chosen wisely in order to produce the tightest possible upper bound. To this end, we introduce the following two qualities of a good genie.
\subsubsection{Useful Genie}
Obtaining tight outer bounds on the capacity region of  multiuser Gaussian channels is generally hindered by the fact that we cannot assume a simple structure (e.g., Gaussian) for the interference seen from other users.  One way around this problem  is to let a genie provide side information to the receivers in such a way that outer bounds can be derived for the genie-aided channel. In the context of the two-user interference channel of interest in this paper, we call a genie {\em useful}, if the sum capacity of the genie-aided channel can be derived.
%In this paper, we also want the Gaussian codes to achieve the sum capacity by treating interference as noise.
An example of useful genie is one that provides side information $S_1 = X_2$ to receiver 1 and side information $S_2 = X_1$ to receiver 2, because the resulting genie-aided channel has no interference. However, being too generous, such a genie does not result in a tight upper bound.
This leads us to the notion of a {smart} genie.
\subsubsection{Smart Genie}
We call a genie smart if it results in a tight upper bound, i.e.,  it should not give too much information to the receivers. The ``smartest" genie, of course, is one that does not interact with the receivers at all; however, it is obviously not useful.

So the essential question is: \emph{Is there a genie that is both useful and smart?}  The question was partly answered in \cite{OneBit2007}, where the genie that results in the upper bound of \eqref{upperbound_onebit} is useful and  {\em asymptotically} smart. What we are looking for is a ``divine genie" that allows us to prove Theorem~\ref{mainresult}.

The quest for the divine genie can be simplified by imposing a structure on the side information it provides. Following \cite{OneBit2007}, we set:
\begin{equation} \label{genie_form}
\begin{split}
S_1 = & \ hX_1 + h\eta{}W_1 \\
S_2 = & \ hX_2 + h\eta{}W_2
\end{split}
\end{equation}
where $\eta$ is a positive real number.
However, unlike in \cite{OneBit2007}, we allow  $W_1$ to be correlated to $Z_1$ (and $W_2$ with $Z_2$), with correlation coefficient  $\rho$.
\begin{lemma}[Useful Genie]
The sum capacity of the genie-aided channel with side information given in \eqref{genie_form} is achieved by using Gaussian inputs and by treating interference as noise if the following condition holds.
\begin{eqnarray}
|h\eta| \leq \sqrt{1-\rho^2}
\label{condition_useful}
\end{eqnarray}
Hence the sum capacity of the symmetric interference channel described  is bounded as
\begin{equation}
\Csum \leq I(X_{1G};Y_{1G},S_{1G}) + I(X_{2G};Y_{2G},S_{2G})
\label{sumcapacity_genieaided}
\end{equation}
\label{useful_genie}
\end{lemma}

\begin{proof}
Using Fano's inequality, we have
\[
\begin{split}
n(R_1 - \epsilon_n) & \leq   I(X_1^n;Y_1^n,S_1^n) \\
& =   I(X_1^n;S_1^n) + I(X_1^n;Y_1^n|S_1^n) \\
& =  h(S_1^n)\! -\! h(S_1^n|X_1^n) \! + \! h(Y_1^n|S_1^n) \! -\! h(Y_1^n|S_1^n,X_1^n)\\
& \stackrel{(a)} =  h(S_1^n) - nh(S_{1G}|X_{1G}) \\
&~~~+ h(Y_1^n|S_1^n) - h(Y_1^n|S_1^n,X_1^n)\\
& \stackrel{(b)} \leq   h(S_1^n) - nh(S_{1G}|X_{1G})\\
&~~~+ nh(Y_{1G}|S_{1G}) - h(Y_1^n|S_1^n,X_1^n)
\end{split}
\]
where step (a) follows from that fact that $h(S_1^n|X_1^n) = h(h\eta{}W_1^n)$ is independent of the distribution of $X_1^n$; and in step (b) we use the facts that 1) the Gaussian distribution maximizes the conditional differential entropy for a given covariance constraint, and 2) the function
\[
h(Y_{1G}|S_{1G}) = \frac{1}{2}\log\left[2\pi\text{e}\left(1 - \rho_1^2 + h^2P_2 + \frac{P_1(\rho_1-\eta_1)^2}{P_1 + \eta_1^2}\right)\right]
  \]
is an increasing and concave function in $P_1$ and $P_2$.
Similarly, we have
\[
\begin{split}
n(R_2 - \epsilon_n) \leq  & \ h(S_2^n) - nh(S_{2G}|X_{2G}) \\
&+  nh(Y_{2G}|S_{2G}) - h(Y_2^n|S_2^n,X_2^n)
\end{split}
%\label{boundonR2}
\]
Thus $n(R_1 + R_2 - 2\epsilon_n) $ is upper bounded by
\[
\begin{split}
& h(S_1^n) - h(Y_2^n|S_2^n,X_2^n) - nh(S_{1G}|X_{1G}) + nh(Y_{1G}|S_{1G}) \\
& + h(S_2^n) - h(Y_1^n|S_1^n,X_1^n)  - nh(S_{2G}|X_{2G}) + nh(Y_{2G}|S_{2G})
\end{split}
\]
Now consider the expression
\[
\begin{split}
h(S_1^n) - h(Y_2^n|S_2^n,X_2^n)= & \ h(hX_1^n+h\eta{}W_1^n) \\
 &~~~- h(hX_1^n+Z_2^n|W_2^n)\\
= & \ h(hX_1^n+V_1^n) - h(hX_1^n+V_2^n)
\end{split}
\]
where $V_1 \sim \mathcal{N}(0,h^2\eta^2)$ and $V_2 \sim \mathcal{N}(0,1-\rho^2)$. Let $V_1$ and $V_2$ be correlated such that $V_2 = V_1 + V$, for some Gaussian random variable $V$ independent of $V_1$, which is possible if $1 - \rho^2 \geq h^2\eta^2$, i.e., \eqref{condition_useful} holds. Thus
\[
\begin{split}
h(S_1^n) - h(Y_2^n|S_2^n,X_2^n) & = h(hX_1^n + V_1^n) - h(hX_2^n + V_2^n) \\
%& =  h(aX_1^n + V_1^n) - h(aX_2^n + V_1^n + V^n) \\
& =  -I(V^n;aX_1^n+V_1^n+V^n) \\
& \stackrel{(a)} \leq  -nI(V;hX_{1G}+V_1+V) \\
& =  nh(S_{1G}) - nh(Y_{2G}|S_{2G},X_{2G}).
\end{split}
\]
where step (a) uses the worst case noise result for the additive noise channel \cite{Diggavi2001}: Gaussian i.i.d. noise with the maximum allowable variance minimizes the mutual information when the input distribution is i.i.d. Gaussian. Therefore $n(R_1 + R_2 - 2\epsilon_n)$ is upper bounded by
\[
\begin{split}
& nh(S_{1G})\! -\! nh(Y_{2G}|S_{2G},X_{2G})\!  +\! nh(S_{2G})\! - \!nh(Y_{1G}|S_{1G},X_{1G}) \\
& \! \! \!\! -nh(S_{1G}|X_{1G})\! +\! nh(Y_{1G}|S_{1G})\! - \! nh(S_{2G}|X_{2G})\! +\! nh(Y_{2G}|S_{2G}) \\
& =  n I(X_{1G};Y_{1G},S_{1G}) + n I(X_{2G};Y_{2G},S_{2G})
\end{split}
\]
and the lemma follows by letting $n \to \infty$ with $\epsilon_n \to 0$.
\end{proof}
\smallskip

\begin{remark}
If the genie does not satisfy (\ref{condition_useful}), it might still be useful. Lemma~\ref{useful_genie} only claims the  `\emph{if}' part, and not the `\emph{only if}' part.
\end{remark}
\smallskip

\begin{lemma}[Smart Genie]
If Gaussian inputs are used, the interference is treated as noise, and the following condition holds
\begin{equation}
\eta\rho = 1+h^2P
\label{condition_smart}
\end{equation}
then the genie does not increase the achievable sum rate, i.e.,
\begin{equation}
\begin{split}
I(X_{1G};Y_{1G}, S_{1G}) = & \ I(X_{1G};Y_{1G}) \\
I(X_{2G};Y_{2G}, S_{2G}) = & \ I(X_{2G};Y_{2G})
\end{split}
\end{equation}
\label{smart_genie}
\end{lemma}
\begin{proof} Since
\[
I(X_{1G};Y_{1G}, S_{1G}) = I(X_{1G};Y_{1G}) + I(X_{1G};S_{1G}|Y_{1G})
\]
we need to determine when $I(X_{1G};S_{1G}|Y_{1G}) = 0$.
Now,
\[
\begin{split}
 I(X_{1G};S_{1G}|Y_{1G})& \\
& \hspace*{-0.1in}= I(X_{1G};X_{1G} + \eta{}W_1|X_{1G} + hX_{2G} + Z_{1G})
\end{split}
\]
Hence $I(X_{1G};S_{1G}|Y_{1G}) = 0$,  if $\eta{}W_1$ is a degraded version of $hX_{2G} + Z_{1G}$,  i.e., if
\[
\mathbf{E}[(\eta{}W_1)(hX_{2G} + Z_1)] =  \mathbf{E}[(hX_{2G} + Z_1)(hX_{2G} + Z_1)]
\]
which happens when $\eta\rho =  1+h^2P$.
\end{proof}
\smallskip

The genie is smart and useful if it meets the conditions of both Lemma \ref{useful_genie} and Lemma~\ref{smart_genie}, i.e., when there exists a $\rho \in [0,1]$ such that
\[
|h+h^3P| \leq |\rho|\sqrt{1-\rho^2}
\]
 which is possible  if
 \[
 |h+h^3 P| \leq .5
 \]
  This completes the proof of Theorem~\ref{mainresult}.

%As a consequence of lemmas \ref{useful_genie} and \ref{smart_genie}, there exists smart and useful genie such that the sum capacity of the genie-aided channel is achieved by treating interference as noise.
%-------------------------------------------------------------------------
\subsection{Geometric Interpretation}
%What is the speciality of the genie used to prove theorem \ref{mainresult}?\\
We now provide a geometric interpretation of the construction of the genie that was used in proving Theorem~\ref{mainresult}. We begin with an evaluation of the mutual information terms on the  RHS of \eqref{sumcapacity_genieaided}. The  term $I(X_{1G};Y_{1G},S_{1G})$ can be expressed as
\[
\begin{split}
I(X_{1G};Y_{1G},S_{1G}) & = I(X_{1G};X_{1G}\!+\!hX_{2G}\!+\!Z_1,h X_{1G}\!+\!h \eta{}W_1)\\
&= I(X_{1G};X_{1G}+hX_{2G}+Z_1,X_{1G}+\eta{}W_1)
 \end{split}
 \]
 which is the mutual information between a  Gaussian random variable and two observations of this random variable in correlated Gaussian noise. The following lemma leads to a geometric interpretation of this mutual information.
\begin{lemma} \label{geom_lemma}
Let $E_i = X_{G} + N_i$, $i = 1\ldots m$, be noisy observations of a zero-mean Gaussian random variable $X_{G}$ with variance $P$, where the variables $N_i$ are arbitrary correlated zero mean Gaussian random variables. Then
\[
I(X_{G},\underline{E}) = \frac{1}{2} \log\left(1 + \frac{P}{\sigma^2}\right)
\]
where $\underline{E} = [E_1 \ldots E_m]^\top$ and
 \[
 \sigma^2 = \inf_{\underline{b}: \ \sum_{i = 1}^{m}b_i = 1}\mathbf{E} \left[(\underline{b}^\top \underline{E} - X_{G})^2 \right]
 \]
\end{lemma}
\smallskip

The proof of the lemma is relegated to the Appendix. A geometric interpretation of the lemma is provided in  Fig.~\ref{fig:geom_lemma}.
\begin{figure}
\centering
\includegraphics[width=0.5\textwidth]{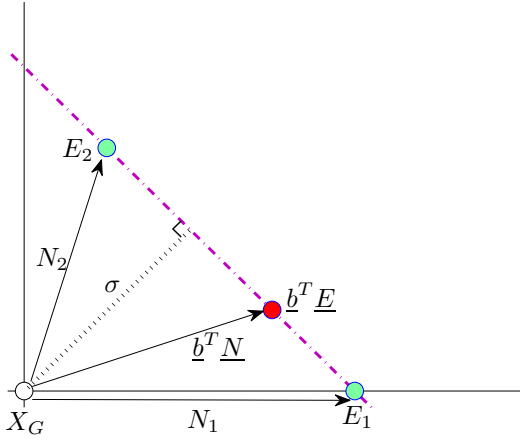}
\caption{The random variables $E_1$ and $E_2$ are represented as vectors with $X_G$ at the origin. The random variable $\underline{b}^\top \underline{E}$ is on the line joining $E_1$ and $E_2$, and $\sigma$  is distance of this line from the origin.}
\label{fig:geom_lemma}
\end{figure}
Specializing Lemma~\ref{geom_lemma} to the case $m=2$ we get the following result for the mutual information term on the RHS of \eqref{sumcapacity_genieaided}.
\begin{lemma}
\[
I(X_{1G};Y_{1G},S_{1G}) = \frac{1}{2}\log\left(1+\frac{P}{\sigma^2}\right)
\]
 where $\sigma$ is the distance from origin to the line joining the points $Q_{Y_1}$ and $Q_{S_1}$ corresponding to $Y_{1G}$ and $S_{1G}$. In polar coordinates (see Fig.~\ref{polar}),
\[
\begin{split}
Q_{Y_1} & =  (\sqrt{1+h^2P},0) \\
Q_{S_1} &  =  (\eta,\theta)
\end{split}
\]
where $\cos\theta$ is the correlation coefficient between $hX_{2G}+Z_1$ and $\eta{}W_1$, i.e.,
\[
\begin{split}
\cos\theta & =  \frac{\mathbf{E}[W_1(hX_{2G}+Z_1)]}{\sqrt{\mathbf{E}[(hX_{2G}+Z_1)(hX_{2G}+Z_1)]}}\\
& =  \frac{\rho}{\sqrt{1+h^2P}}
\end{split}
\]
\label{mvue_line}
\end{lemma}
\begin{remark}
($\eta,\theta$) is an alternate description of the genie that is equivalent to the description ($\eta,\rho$).
\end{remark}
The conditions for the genie to be useful (\ref{condition_useful}) and smart (\ref{condition_smart}) can be transformed into the following conditions (\ref{condition_useful_polar}) and (\ref{condition_smart_polar}),  respectively.
\begin{itemize}
\item {\em Useful Genie: } The genie is useful, if the $(\eta,\theta)$ lies inside the dashed curve in Fig.~\ref{polar}. This region is specified by
\begin{equation}
h^2\eta^2 + (1+h^2P)\cos^2\theta \leq 1
\label{condition_useful_polar}
\end{equation}
\item {\em Smart Genie: } From Lemma \ref{mvue_line}, the genie is smart if $(\eta,\theta)$ lies on the line parallel to y-axis passing through the point $(\sqrt{1+h^2P},0)$, i.e., if
\begin{equation}
\eta\cos\theta = \sqrt{1+h^2P}
\label{condition_smart_polar}
\end{equation}
\end{itemize}
There exists a useful and smart genie if the region specified by (\ref{condition_useful_polar}) intersects with that specified by (\ref{condition_smart_polar}), which is true if (\ref{condition_mainresult}) holds.
\begin{figure}
\centering
\includegraphics[width=0.5\textwidth]{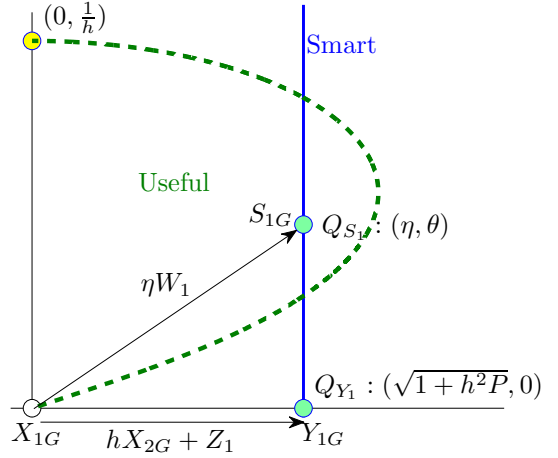}
\caption{The genie is a) useful if it lies inside the dashed curve,  and b) smart if it lies on the solid line. If the dashed curve and solid line intersect, treating interference as noise achieves sum capacity.}
\label{polar}
\end{figure}
%-------------------------------------------------------------------------
\subsection{Upper bound when (\ref{condition_mainresult}) does not hold}
The importance of the geometric intuition will be more evident when the condition (\ref{condition_mainresult}) is not met, i.e., when  the solid line and the dashed curve do not intersect in Fig.~\ref{polar}. In this case, it is of interest to pick the best genie within the class specified in \eqref{genie_form}. The following theorem uses such a genie to obtain an upper bound on the sum capacity .
%in the setting where  $|h +h^3 P| > .5$.
\begin{theorem}
If  $|h +h^3 P| > .5$
\begin{equation}
\Csum \leq \log\left[1 + \frac{P}{1+h^2P}\left(1 + \frac{1}{\mu^2}\right)\right]
\label{upperbound_tangent}
\end{equation}
where $\mu$ is the slope of the tangent from $(\sqrt{1+h^2P},0)$ to the curve (\ref{condition_useful_polar}).
\end{theorem}
\begin{proof}
As illustrated in Fig.~\ref{fig:tangent}, we
choose the genie corresponding to the point where the tangent touches the curve. Let $y = \mu x + c$ be equation of the tangent. Since the line passes through $(\sqrt{1+h^2P},0)$, we have $c^2 = \mu^2(1+h^2P)$. The distance $\sigma$ from origin to the tangent satisfies
\[
\sigma^2 = \frac{c^2}{\mu^2 + 1} = (1+h^2P)\frac{\mu^2}{\mu^2+1}
\]
Thus, by Lemma \ref{mvue_line}, the result follows. \end{proof}
\begin{figure}
\centering
\includegraphics[width=0.5\textwidth]{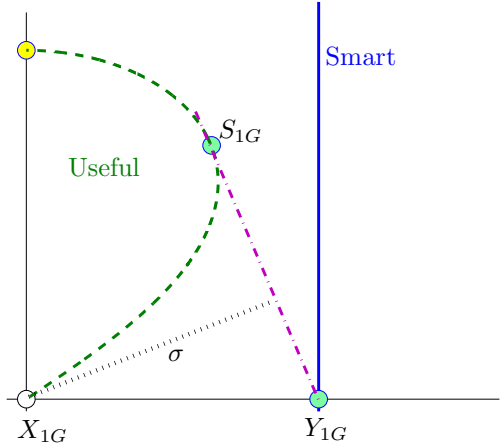}
\caption{Geometric derivation of the upper bound on the sum capacity when (\ref{condition_mainresult}) does not hold.}
\label{fig:tangent}
\end{figure}

%-------------------------------------------------------------------------
%-------------------------------------------------------------------------
\section{Asymmetric Interference Channel}
\begin{theorem}
For the asymmetric interference channel with interference parameters $h_{12}$ and $h_{21}$, suppose  there exist  $\rho_1 \in [0,1]$ and $\rho_2 \in [0,1]$ such that
\begin{equation}
\begin{split}
|h_{12}(1+h_{21}^2P_1)| & \leq  \rho_2\sqrt{1-\rho_1^2} \\
|h_{21}(1+h_{12}^2P_2)| & \leq  \rho_1\sqrt{1-\rho_2^2}
\label{condition_asym}
\end{split}
\end{equation}
Then treating interference as noise achieves sum capacity, which is given by
\[
\Csum = \frac{1}{2}\log\left(1+\frac{P_1}{1+h_{12}^2P_2}\right)  + \frac{1}{2}\log\left(1+\frac{P_2}{1+h_{21}^2P_1}\right)
\]
\label{sum_capacity_asym}
\end{theorem}
\begin{proof}
The proof is similar to that for the symmetric interference channel. We set the genie-aided side information as:
\[
\begin{split}
S_1 = & \ h_{21}(X_1 + \eta_1W_1) \\
S_2 = & \ h_{12}(X_2 + \eta_2W_2)
\end{split}
\]
Let $\rho_1$ be the correlation between $Z_1$ and $W_1$ (and $\rho_2$ the correlation between $Z_2$ and $W_2$). Using the same arguments as in Lemma~\ref{useful_genie}, the genie is useful if
\[
\begin{split}
|h_{21}\eta_1| \leq & \ \sqrt{1-\rho_2^2} \\
|h_{12}\eta_2| \leq & \ \sqrt{1-\rho_1^2}
\end{split}
\]
Also, as in Lemma \ref{smart_genie}, the genie is smart if
\[
\begin{split}
\eta_1\rho_1 = & \ 1+h_{12}^2P_2 \\
\eta_2\rho_2 = & \ 1+h_{21}^2P_1
\end{split}
\]
\end{proof}
\begin{remark}
The condition \eqref{condition_asym} is equivalent to
\begin{equation}
|h_{12}(1+h_{21}^2P_1)| + |h_{21}(1+h_{12}^2P_2)| \leq 1
\label{condition_asym_simple}
\end{equation}
\end{remark}
\smallskip

\begin{proof}
 Set $\rho_1 = \cos\phi_1$ and $\rho_2 = \cos\phi_2$. Then
 \[
 \rho_2\sqrt{1-\rho_1^2} + \rho_1\sqrt{1-\rho_2^2} = \sin(\phi_1 + \phi_2) \leq 1
 \]
 Thus \eqref{condition_asym} implies \eqref{condition_asym_simple}. On the other hand, if \eqref{condition_asym_simple} is satisfied, we can find $\phi$ such that
 \[
 |h_{12}(1+h_{21}^2P_1)| \leq cos^2\phi \leq 1 - |h_{21}(1+h_{12}^2P_2)|
 \]
 i.e.,
 \[
 \begin{split}
 |h_{12}(1+h_{21}^2P_1)| \leq & \cos^2\phi \\
 |h_{21}(1+h_{12}^2P_1)| \leq & \sin^2\phi
 \end{split}
 \]
 Setting $\rho_1 = \sin\phi$ and $\rho_2 = \cos\phi$, we have \eqref{condition_asym}.
\end{proof}
\begin{remark}
The sum capacity of the one-sided interference channel \cite{Kramer2004}, \cite{Sason2004}, \cite{OneBit2007} and hence the Z-channel outer bound on the sum capacity of the symmetric interference channel \eqref{upperbound_kramer} are immediate corollaries of Theorem~\ref{sum_capacity_asym}.
\end{remark}
%\begin{remark}
%When the conditions of (\ref{condition_asym}) are not met, upper bounds on the sum capacity  can be derived in a similar fashion as in the symmetric case. The details are omitted here.
%\end{remark}
%-------------------------------------------------------------------------
%-------------------------------------------------------------------------
\section{Conclusion}
We used a genie-aided channel to derive new upper bounds on the sum capacity of the two-user Gaussian interference channel. We introduced the notions of {\em useful genie} and {\em smart genie}. A genie is useful if the sum capacity of the genie-aided channel can easily be derived, and smart if the sum capacity of the genie-aided channel is the same as that of the interference channel. We showed that when the interference levels are below certain thresholds, we can construct a genie that is both useful and smart. Thus we established the sum capacity of the interference channel in the low interference regime, and furthermore showed that it is optimal for the receivers to treat the interference as noise in this regime. We were recently informed by G. Kramer that Theorem~\ref{sum_capacity_asym} has been independently established in \cite{Shang-Kramer-Chen-2007,Motahari-Khandani-2007}.

The notion of a useful and smart genie is generalizable to interference channels with more than two users. We are currently working on establishing sum capacity results for such interference channels.

%-------------------------------------------------------------------------
%-------------------------------------------------------------------------
\section*{Acknowledgment}
This research was supported in part by the NSF award CCF 0431088, through the University of Illinois, by a Vodafone Foundation Graduate Fellowship, and a grant from Texas Instruments.
%-------------------------------------------------------------------------
%-------------------------------------------------------------------------

\section*{Appendix}

{\em Proof of Lemma~\ref{geom_lemma}:}
From Data processing inequality, it follows that
\[
I(X_{G}; \underline{E}) \geq  I(X_{G};\underline{b}^\top\underline{E}), ~~\forall \underline{b}
\]
i.e., that
\[
I(X_{G}; \underline{E}) \geq  \sup_{\underline{b}: \ \sum_{i = 1}^{m}b_i = 1} I(X_{G};\underline{b}^\top\underline{E})
\]
Since $X_{G}$ and $\underline{N}$ are Gaussian, the minimum mean squared-error (MMSE) estimator of the random variable $X_{G}$ based on $\underline{E}$  is a linear function of $\underline{E}$ and is also a sufficient statistic. Hence $I(X_{G},\underline{E}) = I(X_{G};\underline{b}^\top \underline{E})$ for some $\underline{b}$. Therefore,
\[
I(X_{G},\underline{E}) = \sup_{\underline{b}: \ \sum_{i = 1}^{m}b_i = 1} I(X_{G};\underline{b}^\top\underline{E})
\]
and the lemma follows.

\bibliographystyle{IEEE}
\bibliography{ITA2008}

\begin{thebibliography}{10}
\providecommand{\url}[1]{#1}
\csname url@rmstyle\endcsname
\providecommand{\newblock}{\relax}
\providecommand{\bibinfo}[2]{#2}
\providecommand\BIBentrySTDinterwordspacing{\spaceskip=0pt\relax}
\providecommand\BIBentryALTinterwordstretchfactor{4}
\providecommand\BIBentryALTinterwordspacing{\spaceskip=\fontdimen2\font plus
\BIBentryALTinterwordstretchfactor\fontdimen3\font minus
  \fontdimen4\font\relax}
\providecommand\BIBforeignlanguage[2]{{%
\expandafter\ifx\csname l@#1\endcsname\relax
\typeout{** WARNING: IEEEtran.bst: No hyphenation pattern has been}%
\typeout{** loaded for the language `#1'. Using the pattern for}%
\typeout{** the default language instead.}%
\else
\language=\csname l@#1\endcsname
\fi
#2}}

\bibitem{Carleial1975}
A.~B. Carleial, ``{A case where interference does not reduce capacity},''
  \emph{IEEE Trans. on Inform. Theory}, vol. IT-21, no.~1, pp. 569--570, Sept.
  1975.

\bibitem{HK1981}
T.~S. Han and K.~Kobayashi, ``{A new achievable rate region for the
  interference channel},'' \emph{IEEE Trans. on Inform. Theory}, vol. IT-27,
  no.~1, pp. 49--60, Jan. 1981.

\bibitem{Sato1981}
H.~Sato, ``{The capacity of the Gaussian interference channel under strong
  interference},'' \emph{IEEE Trans. on Inform. Theory}, vol. IT-27, no.~6, pp.
  786--788, Nov. 1981.

\bibitem{Kramer2004}
G.~Kramer, ``{Outer bounds on the capacity region of Gaussian interference
  channels},'' \emph{IEEE Trans. on Inform. Theory}, vol. IT-50, no.~3, pp.
  581--586, March 2004.

\bibitem{Sason2004}
I.~Sason, ``{On the achievable rate regions for the Gaussian interference
  channel},'' \emph{IEEE Trans. on Inform. Theory}, vol. IT-50, no.~6, pp.
  1345--1356, June 2004.

\bibitem{OneBit2007}
R.~H. Etkin, D.~N.~C. Tse, and H.~Wang, ``{Gaussian interference channel
  capacity to within one bit},'' \emph{Submitted to IEEE Trans. on Inform.
  Theory}, Feb. 2007.

\bibitem{Costa1985}
M.~H.~M. Costa, ``{On the Gaussian interference channel},'' \emph{IEEE Trans.
  on Inform. Theory}, vol. IT-31, no.~5, pp. 607--615, Sept. 1985.

\bibitem{CMG2007}
H.~F. Chong, M.~Motani, H.~K. Garg, and H.~E. Gamal, ``{On the Han-Kobayashi
  region for the interference channel},'' \emph{Submitted to IEEE Trans. on
  Inform. Theory}, Aug. 2006.

\bibitem{Carleial1978}
A.~B. Carleial, ``{Interference channels},'' \emph{IEEE Trans. on Inform.
  Theory}, vol. IT-24, no.~1, pp. 60--70, Sept. 1978.

\bibitem{Diggavi2001}
S.~Diggavi and T.~M. Cover, ``{Worst additive noise under covariance
  constraints},'' \emph{IEEE Trans. on Inform. Theory}, vol. IT-47, no.~7, pp.
  3072--3081, Nov. 2001.

\bibitem{Shang-Kramer-Chen-2007}
X.~Shang, G.~Kramer, and B.~Chen, ``{A new outer bound and noisy-interference
  sum-rate capacity for the Gaussian interference channels},'' \emph{Submitted
  to IEEE Trans. on Inform. Theory}, Dec. 2007.

\bibitem{Motahari-Khandani-2007}
A.~S. Motahari and A.~K. Khandani, ``{Capacity bounds for the Gaussian
  interference channel},'' \emph{To be submitted to IEEE Trans. on Inform.
  Theory}.

\end{thebibliography}

\end{document}